\documentclass[conference]{IEEEtran}

\usepackage[utf8]{inputenc}
\usepackage{paralist}
\usepackage{mathtools}
\usepackage{comment}
\usepackage{cases}
\usepackage{etoolbox}
\newtoggle{full}  
\newtoggle{long}  
\toggletrue{long}
\iftoggle{full} {     \usepackage{xcolor}
  \specialcomment{longvers}{\begingroup\color{orange}\footnotesize}{\endgroup}
  \specialcomment{sixpages}{\begingroup\color{red}}{\endgroup}
}{
    	\iftoggle{long} {
	  \includecomment{longvers}
	  \excludecomment{sixpages}
	}{ 	  
		\excludecomment{longvers}
	  \includecomment{sixpages}
	  	  
	}
}

\usepackage{algpseudocode}  
\usepackage[inline]{enumitem}
\usepackage{tabto}
\usepackage{subcaption}

\usepackage{amsthm}
\usepackage{url}

\newtheorem{Mytheorem}{Theorem}
\newtheorem{Mylemma}{Lemma}

\newtheorem{algorithm}{Algorithm}

\DeclarePairedDelimiter{\ceil}{\lceil}{\rceil}
\DeclarePairedDelimiter{\floor}{\lfloor}{\rfloor}

\newcommand{\hash}{\textrm{hash}}

\begin{document}

		\title{Scaling Blockchains \\ Without Giving up Decentralization and Security \\ {\large A Solution to the Blockchain Scalability Trilemma}}

\author{
	\IEEEauthorblockN{Gianmaria Del Monte}
	\IEEEauthorblockA{\textit{Engineering Dept.} \\
		\textit{Roma Tre University}\\
		Roma, Italy \\
		gia.delmonte@stud.uniroma3.it}
	\and
	\IEEEauthorblockN{Diego Pennino}
	\IEEEauthorblockA{\textit{Engineering Dept.} \\
		\textit{Roma Tre University}\\
		Roma, Italy \\
		pennino@ing.uniroma3.it\\
		0000-0001-5339-4531}
	\and
	\IEEEauthorblockN{Maurizio Pizzonia}
	\IEEEauthorblockA{\textit{Engineering Dept.} \\
		\textit{Roma Tre University}\\
		Roma, Italy \\
		pizzonia@ing.uniroma3.it\\
		0000-0001-8758-3437}
}

	\maketitle

\begin{abstract}
	
	Public blockchains should be able to scale with respect to the number of nodes
	and to the transactions workload. The \emph{blockchain scalability trilemma}
	has been informally conjectured. This is related to scalability, security and
	decentralization, stating that any improvement in one of these aspects should
	negatively impact on at least one of the other twos.
	In fact, despite the large research and experimental
	effort, all known approaches turn out to be tradeoffs.
	
	We theoretically describe a new blockchain architecture that scales to arbitrarily high
	workload provided that a corresponding proportional increment of nodes is
	provisioned. We show that, under reasonable assumptions, our approach does not 
	require  tradeoffs on security or decentralization.
	To the best of our knowledge, this is the first
	result that disprove the trilemma considering the scalability of all 
	architectural elements of a blockchain and not only the consensus protocol. 
	While our result is currently only theoretic, we believe
	that our approach may stimulate significant practical contributions.
	
\end{abstract}


\section{Introduction}

Scalability of public blockchains is extremely important for their
success. All communities and companies engaged in the development of public
blockchains strive to create solutions to support a large number of nodes and/or
large workloads (in terms of transactions per second). While distributed systems
or peer-to-peer systems can scale very well with respect to these parameters,
blockchains, till now, have represented a very hard challenge. The difficulty
comes from the interplay of contrasting requirements in the blockchain design.
Vitalik Buterin, founder of the Ethereum project, summarized his understanding
about this problem by introducing the so-called \emph{blockchain scalability
trilemma}. This trilemma states that regarding scalability, security and
decentralization, any improvement in one of these aspects negatively impacts on
at least one of the other twos. 

Intuitively, consider an ideal blockchain, in which nodes are all equal
(with same and constant cpu, bandwidth, storage, etc.). If we double both load and 
number of nodes, the ratio between load and available
processing resources is constant, and scalability might, in principle, be possible.  
The obvious question is if it is
possible to design a blockchain in which all activities (consensus, storage, and
communication) are distributed (and not just replicated) across nodes, keeping 
both high security and high decentralization, independently from the load and the number of nodes.

Current blockchains broadcast all pending transactions and accepted blocks to all
nodes. Hence, each node is a scalability bottleneck, since its
processing power and its bandwidth are bounded. While a large literature addresses the 
 scalability of the consensus algorithm, this is only a part of the story.
The \emph{sharding} approach, propose to 
partition the whole network (comprising nodes, blocks, consensus, and
transaction history) into several smaller networks (\emph{shards}), with
some form of coupling among them (to handle inter-shard transactions). 
However, shards have less nodes than the whole network, reducing decentralization and security,
and inter-shard transactions may pose a scalability challenge.

In this paper, we theoretically describe a novel blockchain architecture for
which the maximum workload that can be processed is proportional to the number
of nodes involved in the blockchain. Our scalability result is related to all
architectural aspects, i.e., no node is ever forced to receive, process or store
a quantity of data that is proportional to the whole workload. In our approach,
 decentralization
and security, are not affected when the number of nodes is increased. This derives from 
the adoption, as a building block, of a randomized committee-based
consensus, like the one described in~\cite{chen2019algorand}.

The value of our contribution is twofold. Firstly, it shows that, under
realistic assumptions, it is possible to solve the blockchain scalability
trilemma. Secondly, it provides a construction that can be of inspiration for
practical realizations.

The rest of the paper is structured as follows. In Section~\ref{sec:soa}, we quickly review some related
literature. In Section~\ref{sec:background}, we provide
basic definitions, assumptions and some background. 
\begin{longvers}
In Section~\ref{sec:current_problems}, we focus on scalability problems of current solutions.
\end{longvers}
In Section~\ref{sec:solution}, we first present the main ideas of our architecture and then 
we detail the tasks performed by each committee. 
\begin{longvers}
In Section~\ref{sec:correctness}, we formally prove correctness.
\end{longvers}
In Section~\ref{sec:scalability}, we show a formal scalability result.
In Section~\ref{sec:discussion}, we discuss the effectiveness of our approach and some other aspects.
In Section~\ref{sec:conclusions}, we draw the conclusions.

\section{State of the Art}\label{sec:soa}

Despite its practical relevance and its potential theoretical impact, the
blockchain scalability trilemma~\cite{ethfaqsharding} has not a formal
definition in scientific literature. However, ``the trilemma'' is often cited in
research works. For example, a taxonomy of consensus algorithm based on the scalability
trilemma is provided in~\cite{altarawneh2020buterin} and surveys on blockchain
scalability explicitly cite it~\cite{xie2019survey,zhou2020solutions}.

\begin{longvers}
Concerning scalability, the most interesting results are about 
sharding. However, all of them suffer, in different ways, from the strong
partitioning of the whole blockchain. 
RapidChain~\cite{zamani2018rapidchain} requires strong synchronous communication
among shards which is hard to achieve.
AHL~\cite{dang2019towards} is an approach that turns out to be expensive when
a transaction involves multiple shards.
Omniledger~\cite{kokoris2018omniledger} requires every participant in the
consensus protocol (a subset of the nodes in each shard) to broadcast a message 
to the entire network to verify transactions. It also requires the users 
to participate actively in cross-shard transactions verification.
Elastico~\cite{luu2015scp} and
Monoxide~\cite{wang2019monoxide}, require to execute expensive proof-of-work to
decide which shard should process a transaction.
\end{longvers}
\begin{sixpages}
Concerning blockchain scalability, the most interesting results are about
sharding~\cite{zamani2018rapidchain,dang2019towards,kokoris2018omniledger,wang2019monoxide}.
\end{sixpages}
Other proposals are targeted to scalability of consensus algorithms, like the
EOS network, which is analyzed in~\cite{xu2018eos},  
Algorand~\cite{chen2019algorand}, and Ouroboros~\cite{david2018ouroboros,kiayias2017ouroboros}.
\begin{sixpages}
An in-depth review of the problems of current blockchain solutions is provided in~\cite{scaling-arXiv}.
\end{sixpages}

A fundamental role in blockchain scalability is played by \emph{Authenticated
	Data Structures} (ADS) and related techniques. 
The works~\cite{chen2019algorand}
and~\cite{bpp-bmdhtdvfss-19} propose to use authenticated data structures to
scale with respect to the size of the blockchain state. Other works related to
scalability of ADSes
are~\cite{ppg-pisuadsuc-19,ppp-oiesarqads-ieee-19,gps-usbcheckin-2016,edgp-piaidpics-2017,gp-spuuutd-2016}.

\section{Basic Definitions and Assumptions}\label{sec:background}

In the following, we introduce some definitions and assumptions that are used in the rest of the paper. 

We call \emph{candidate} (or \emph{pending}) transactions those that are generated by users but are
not (yet) processed by the blockchain. A \emph{confirmed} (or \emph{accepted} or \emph{committed}) transaction is a
transaction that was successfully processed by the blockchain. The
\emph{history} of the blockchain is a totally ordered sequence of confirmed
transactions. For simplicity, in the rest of the paper, we mostly focus on a
simple model of blockchain that realizes a pure cryptocurrency. In other words,
each \emph{address} (or \emph{account}) is associated to a wallet with a non
negative balance and a transaction just moves currency from a wallet to another
changing balances, accordingly. While the results of this paper may be
applicable in a more powerful model, for simplicity, we do not consider more
complex cases. The \emph{state} of the blockchain is the balances of all
addresses, at a certain instant. Since confirmed transactions are totally ordered, the
state of the blockchain is defined between two consecutive accepted transactions. The
sequence of confirmed transactions is split into \emph{blocks} that are
sequentially numbered. Pending transactions are confirmed when a new block is
\emph{created} (or \emph{mined}). The mining of a new block requires to (1)
select a subset of pending transactions, (2) order them, and (3) verify that
certain \emph{consensus rules} are fulfilled when a transaction is applied to
the previous state. In practice, consensus rules may be complex, but, in our
model, they are limited to keeping non-negative balances.

The \emph{nodes} of a blockchain network act as peers in the sense that they
perform the same actions. Commonly, they broadcast candidate transactions (even though, in the approach described in this paper this is not true). Each node keeps
a set of candidate transactions it knows. Nodes jointly perform a \emph{consensus algorithm} (or
\emph{consensus protocol}) to reach \emph{consensus} on the transactions to be
included in the next block. This also implies a consensus on the state after that
block (i.e., after the last transaction of the block). We assume 
that consensus does not introduce forks as in~\cite{chen2019algorand}.
In our approach, the transactions to be included in the next block are 
decided by a joint work of a number of \emph{committees}, each of them 
performing a consensus algorithm. This paper describes in detail 
how this can be done. 

We consider all nodes to have always the same constant amount of resources in
terms of CPU, storage, and network bandwidth. For simplicity, we also assume
that honest nodes are reliable, that all communications between nodes are
instantaneous, and that the backbone of the underlying network does not
introduce any bottleneck and is reliable.

The stream of candidate transactions is the \emph{workload} (or simply
\emph{load}) of the blockchain and its magnitude is a frequency measured in
transactions per seconds. In the following, we assume that accounts whose
balance is changed by candidate transactions are uniformly distributed
on the address space.
The time a candidate transaction takes to be confirmed
is called \emph{(confirmation) latency}. The \emph{maximum throughput} of the
blockchain is the frequency of candidate transactions that it is possible to
confirm with bounded latency, i.e., avoiding indefinitely growth of the set of
pending transactions. When workload is less than the maximum throughput, we say
that the blockchain is \emph{well-provisioned} (for that workload).
Since we are interested in investigating scalability of blockchains, we compare situations in which the same blockchain architecture is adopted with different number of
nodes and different workload. 
We say that a blockchain architecture \emph{scales} when starting
from a well-provisioned blockchain, it is possible to give a new well-provisioned 
blockchain with proportionally increased load and nodes. For simplicity,
we do not deal with dynamically changing number of nodes.

\begin{sixpages}
We also assume general knowledge of \emph{Merkle trees} and related concepts, like \emph{proofs}, \emph{root-hash} and \emph{pruning} (see~\cite{scaling-arXiv} and~\cite{bpp-bmdhtdvfss-19}).
\end{sixpages}
\begin{longvers}
In the rest of the paper, we also assume knowledge of \emph{Merkle trees}. A
Merkle tree $T$ is a complete binary tree in which each leaf contains a balance of an address and each internal node contains the \emph{cryptographic hash}
of the content of its children. The root of $T$ contains the \emph{root-hash} of $T$. From $T$,
it is possible to provide a (logarithmic length) proof of the value of each of
the leaf $v$ of $T$, by providing the content of the siblings of the nodes on the path from $v$ to the root. The proof can be verified  only against the root-hash of $T$ from which it is
derived and it is hard for an attacker to synthesize a valid proof.
If proofs are needed for only a subset $S$ of the leaves of $T$, it is possible to prune
$T$ so that size of $T$ is proportional to $|S|$ and it is still possible to obtain proofs for each element of $S$. Further details about
the use of pruned Merkle trees in blockchains can be found in~\cite{bpp-bmdhtdvfss-19}.

\end{longvers}

\begin{longvers}

\section{Problems of Current Approaches}\label{sec:current_problems}

In this section, we list some relevant aspects that make current common public
blockchains architecture not scalable (according to the scalability definition
given in Section~\ref{sec:background}).

In this paper we focus on the following problems.
\begin{enumerate}
	
	\item New candidate transactions are always broadcasted to all
	(validating/mining) nodes. 
	
	\item There exists a set of (validating/mining) nodes (possibly comprising all nodes), each processing all candidate transactions that have to be included in the next block.
	
	\item Each new block is broadcasted to all nodes. 
	
\end{enumerate}

Each of these aspects implies that, to well-provision the blockchain, individual
nodes have to increase computing power and bandwidth even under proportionality
condition.

We purposely avoid to mention scalability problems related to the computational complexity of the consensus
protocol, since these three aspects are independent from it and are relevant
even for blockchains that adopt light consensus protocols. We also avoid to
mention the problem of storing the whole blockchain state in each node, which is
already addressed in other works~\cite{bpp-bmdhtdvfss-19,leung2019vault}.

In literature, most of the proposals that address the above scalability problems
introduce some form of \emph{sharding}, which is a way to partition the
blockchain network and the transactions, in effect, creating a sort of federation of a multiplicity of blockchain networks. The sharding technique
suffers of a number of problems. The most hard-to-solve ones derives from the
fact that state is partitioned across shards. Hence, if shards are many, most
transactions turn out to change the state of more than one shard. These are called 
\emph{inter-shard} transactions. Clearly, each transaction should be atomic, which is
not so simple to achieve in a sharded environment. Typically this introduces
inefficiencies related to inter-shard communication (usually performed using
some form of broadcast) and the need of techniques similar to a two-phase commit
to ensure that all transactions executions are atomic. 
Another strong criticism is about security, since smaller shards are supposed to allow for better scalability but
are deemed to be less secure than larger ones. 

The contribution of this paper is the description of an architecture that aims
at addressing the three above-mentioned scalability problems. We do this without
relying on sharding, in the sense, in our case, the blockchain is one.
However, we introduce a way of dynamically sharing the load among
nodes. Our solution intends to apply a parallel version of the Algorand consensus
approach~\cite{chen2019algorand}. We also leverage the idea of distributing the storage of the state, as
described in~\cite{bpp-bmdhtdvfss-19}.

\end{longvers}

\section{A Scalable Blockchain Architecture}\label{sec:solution}

In this section, we describe an architecture that achieves scalability, as
defined in Section~\ref{sec:background}. We first informally describe 
ideas that make scalability possible in our architecture, then we list in detail 
all the components of the architecture and their behavior.

\subsection{Main Ideas}\label{sec:ideas}

\paragraph{Committees} In our approach, there are a number of \emph{committees}
that collectively work to perform the computation needed to validate and confirm
transactions and to compute the new block. Each committee is made of a number of
nodes called \emph{members}. The way members of each committee are selected is
not important to understand our architecture.  However, for security reasons, a
randomized approach that regularly change committee members can be adopted (like
for example in~\cite{chen2019algorand}). We assume that all committees are equal
sized. Their size is fixed and does not change when the number of nodes in the
network changes. The committees cooperates by exchanging messages. A discussion 
about inter-committee communication and of periodically changing the  committee members is provided in Section~\ref{sec:discussion}.

\paragraph{Blocks} Differently from the common approach, in our architecture,
a block conceptually aggregates transactions, but transactions are never explicitly represented in the broadcasted block. In fact, the transactions 
related to a block are proportional to the workload. Forcing a node to receive
all of them would impair scalability. Instead, for each new block, we only broadcast
constant size data. We call \emph{block} this constant size data. Our block can be considered 
equivalent in content to the block header of other traditional approaches. For our theoretical analysis, it is
only relevant to know that the block contains the hash of the previous block, and the hash of the blockchain state
after the application of all transactions of the block. 
The state hash is computed on the basis of a
Merkle tree, hence we call it \emph{state root-hash}. 
In the following, we explain how the computation of the state root-hash
is shared across several committees. In principle, blocks might also contain a hash of the
transactions of the block. However, this is not strictly needed for 
the correctness of the execution, which computes the new state on the basis of the 
previous one. Hence, we ignore it. Consider also that techniques 
similar to those that we propose for scaling with respect to the computation the state root-hash could be adopted for a
root-hash that summarizes the transactions of the block.

\paragraph{Storage} In our solution, we are interested in storing the state of the blockchain (a similar approach can be adopted for transaction history, but we do not include it in our model). For scalability reasons, it is not possible for all nodes to store 
the whole state. This is not only because of the size of the needed storage, but also because processing updates to the whole state would require
an amount of resources proportional to the workload. Bernardini et
al.~\cite{bpp-bmdhtdvfss-19} and Vault~\cite{leung2019vault} propose approaches
that do not require for all nodes to store the whole blockchain state. In these approaches, a node
may not even store any state at all and still be able to participate in transactions confirmation 
and block creation activities. In the
following, we refer to a node that stores a part of the blockchain state as a \emph{storage
	node}. The cited works, consider a (complete and binary) Merkle
tree on the whole address space in which each leaf is an address (comprising unused
addresses). We denote this Merkle tree by $W$.
Storage nodes store only a part of the state (a subset of all accounts) and the corresponding part of 
$W$, that covers all the paths from stored addresses to the root, pruning the rest of $W$ (see details in~\cite{bpp-bmdhtdvfss-19}). 
The root-hash of $W$ is the state root-hash.  Since blockchain state changes, the block contains the state root-hash of the state after the application of the last transaction of the block.

\paragraph{Transaction creation} As in~\cite{bpp-bmdhtdvfss-19}
and~\cite{leung2019vault}, in our approach, a node $n$ that intends to create a transaction has
the responsibility to provide cryptographic proofs of the balances of the
accounts that are going to be changed by the transaction (i.e., of the accounts \emph{involved} in the transaction). These cryptographic proofs are asked by $n$ to one or more storage nodes. Since
each storage node stores a pruned Merkle tree, they are able to provide that proof for the accounts they store.
However, as will be clear in the following, the balances provided by storage nodes
are related to a state that is delayed by a few blocks. The proof $p$ obtained from a storage node is \emph{related} to 
a state of the blockchain after a certain block $B$, intending
that its is valid with respect to the state root-hash in $B$. 
 We also simply write that $p$ is related to $B$.
Since nodes keep only a truncated list of blocks (see below), proofs that are too old
cannot be validated and we say that they are \emph{expired}.

\begin{figure*}
\begin{subfigure}{0.8\textwidth}
	\centering
	\includegraphics[width=\linewidth]{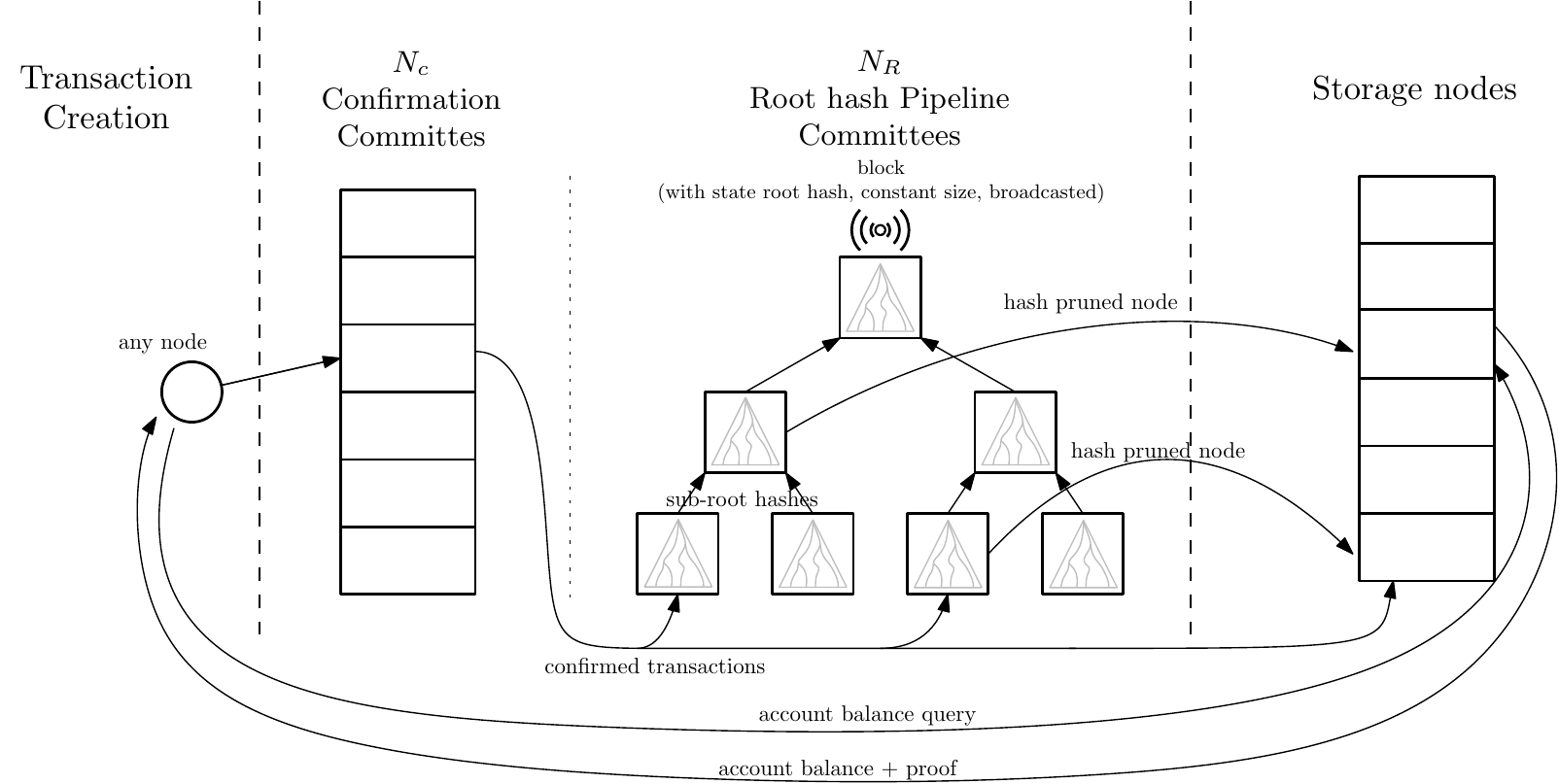}
	\subcaption{} 
	\label{fig:architecture}
\end{subfigure}
\hfill
\begin{subfigure}{0.14\textwidth}
	\centering
	\includegraphics[height=\linewidth,angle=90]{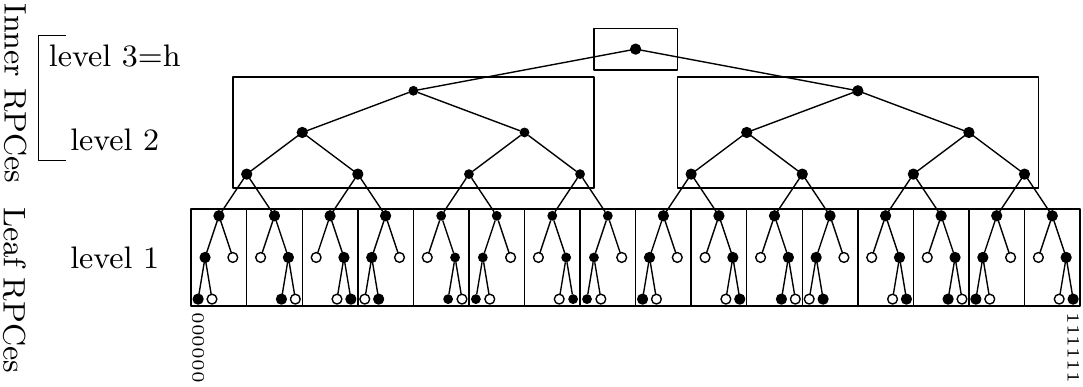}
    \subcaption{}
	\label{fig:WpartitionedByRPCs}
\end{subfigure}
\caption{In~\ref{fig:architecture}, the flow of data within the proposed architecture. In~\ref{fig:WpartitionedByRPCs}, The tree conceptual Merkle tree $W$ partitioned into underlying trees of  RPCes. White nodes are pruned.}
\end{figure*}

\begin{longvers}
\begin{figure*}
	\centering
	\includegraphics[width=\linewidth]{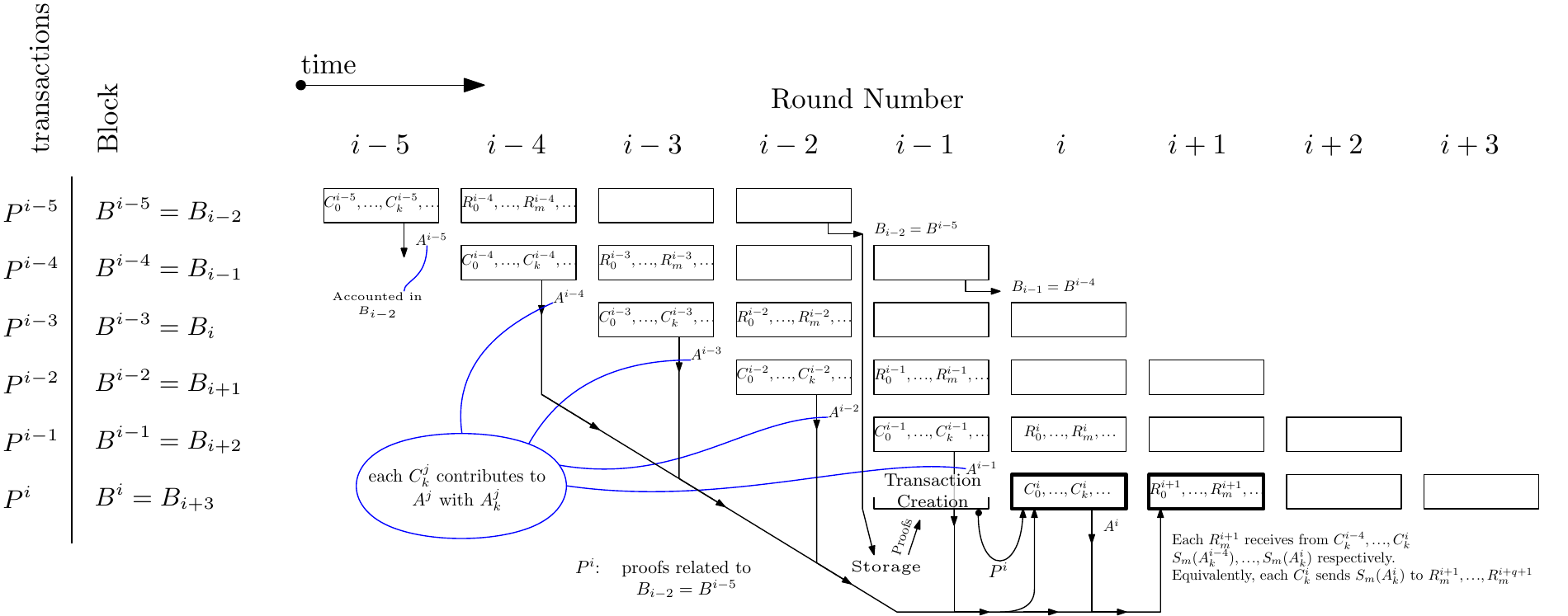}
	\caption{An example of execution of a pipeline with four stages. In the picture, inputs to $C_k^i$ and to $R_m^{i+1}$ are evidenced.} 
	\label{fig:pipeline}
\end{figure*}
\end{longvers}

\paragraph{Pipelining} Consider the computation needed to validate and confirm
transactions and then to compute the new state root-hash to be included in a new
block. The effort needed for this task is clearly proportional to the workload.
To scale, it is important to distribute  this computation 
over several committees. For this reason, we introduce a pipeline in which the computation 
is performed in several \emph{stages}. In each stage, the computation can be further distributed across several committees, 
using a parallel processing approach. The output of the last stage is a new block.

We suppose that the time is spliced into equal length \emph{rounds}.
Rounds are sequentially numbered. In each round, each committee performs its
task for a certain pipeline stage. The result of the computation of a committee is communicated to the committees that, in the next round, need it for the next
pipeline stage.  For simplicity, we assume that
all communications are instantaneous. A discussion on communication problems is provided
in Section~\ref{sec:discussion}.

We denote by $B_i$ the block produced as output of the last stage in round $i$. We
denote the block that contains the transactions that entered the pipeline at round $i$ by
$B^i$. If the pipeline has $q$ stages, the transactions that enter the pipeline
at round $i$, and that are accepted, will be part of the block produced as
output by the last stage that runs at round $i+q-1$. Hence, we have that $B^i=B_{i+q-1}$. 
The first round in which $B^i$ can be used by any node is round $i+q$.
\begin{longvers}
In Figure~\ref{fig:pipeline}, an example of pipeline with $q=4$ is depicted. Details are explaind in the following.
\end{longvers}

Each produced block is propagated to all nodes. Each storage node replies only
with proofs related to already produced blocks. We assume that a node that
creates a transaction $x$ takes one round to retrieve all the proofs for
balances of addresses involved in $x$. These proofs are validated in the first
stage of the pipeline. Note that, at that time, these proofs are old by at least
$q$ rounds.

\paragraph{Truncated block history} As in~\cite{bpp-bmdhtdvfss-19}, we assume
each node does not keep all the blocks, but only the last $d$ blocks received.
That is, at round $i$, each node stores blocks
$B_{i-1}=B^{i-q},\dots,B_{i-d}=B^{i-q-d+1}$ and previous blocks are
\emph{forgotten}. A proof $p$ related to $B_j$ is \emph{expired} at
round $i$ if $ j < i-d $ (i.e., nodes have forgotten $B_j$ needed to validate
$p$). Since in round $ i $ the last available block is $B_{i-1}$, a storage node replies
with proofs related to that block. In our model, we assume a node takes one round to create a transaction (asking for proofs to the storage). Hence, for the nodes participating
in committees of the first pipeline stage to be able to verify balance proofs, it
should be $d \geq 2$. In the following we assume $ d = 2 $. In a practical realization, $ d $ might be larger to compensate delays of the network~\cite{bpp-bmdhtdvfss-19}.

\subsection{The architecture}\label{sec:architecture}

In Figure~\ref{fig:architecture}, we show the
proposed architecture and the flow of information within it. We describe it from
left to right. 

Any node can create a candidate transaction. As described above, a new candidate
transaction should come with balances of the involved accounts and with corresponding 
proofs of integrity, related to a previous round. This can be obtained from a storage node. Candidate transactions are
not broadcasted into the network (nothing is broadcasted in our approach but 
constant size blocks), instead, they are sent to a limited number of nodes as described
below.

The validation of the set of transactions that have to be included in a block
is performed by \emph{Confirmation Committees} (\emph{CC}). We denote each
distinct CC by $C_k$ with $k=1,\dots,N_c$, where $N_c$ is
the number of CCes. When relevant, we write $C_k^i$ intending to denote the
$k$-th confirmation committee that runs in the $i$-th round. The node that
creates a new transaction $x$ sends it to $C_k^i$, with $k =
(\hash(x_{\textrm{src}}) \mod N_c)$, where $x_{\textrm{src}}$ is the account whose
balance is charged by $x$. We intend that $x$ is received by $C_k^i$ before the
start of round $i$ and hence $C_k^i$ can process it during round $i$. We say that
$C_k^i$ is \emph{responsible} for that transaction. The set of candidate 
transactions for which $C_k^i$ is responsible is denoted $P(C_k^i)$. 
We denote by
$P^i=\bigcup_k P(C_k^i)$ the set of candidate transactions processed by all
confirmation committees in round $i$. The result provided by $C_k^i$ is a
\emph{sequence} of transactions denoted $A_k^i$, with $A_k^i \subseteq P(C_k^i)$.

 A fundamental
aspect of the algorithm performed by $C_k^i$ is to obtain, for each transaction $ x $,
the balance of $ x_{\textrm{src}} $ related to $ B^{i-1} $ to check that $x$ complies with the non-negative balance rule. Since
proofs attached to transactions are related to $ B_{i-2} = B^{i-q-1} $, those proofs are for balances that are old.
 In fact, they might have been outdated by transactions accepted in the
last $q$ rounds, for which the corresponding block is not yet available. Hence,
each $C_k^i$ should also be aware of state changes induced by transactions
accepted by $C_k^{i-q}, \dots, C_k^{i-1}$, that is, of $A_k^{i-q}, \dots,
A_k^{i-1}$, 
\begin{sixpages}
respectively.
\end{sixpages}
\begin{longvers}
respectively (see Figure~\ref{fig:pipeline}).
\end{longvers}
These transactions are considered to update all account
balances involved in $ P(C_k^i) $ to match state related to $ B^{i-1} $. We call
\emph{time-updating} this process.
\begin{longvers}
In Figure~\ref{fig:pipeline}, we depicted the pipeline and put in evidence the inputs for a generic $ C^i_k $.
\end{longvers}

In our model, each $C_k^i$ performs the following algorithm (by a suitable consensus
protocol).
\begin{algorithm}[Confirmation]\label{algo:confirmation}
\end{algorithm}
\begin{enumerate}
	
\item It checks that each transaction in $P(C_k^i)$ fulfills syntactic rules and
proofs are not expired. It discards non-compliant transactions, resulting in $P'(C_k^i) \subseteq P(C_k^i)$

\item\label{step:conf:sort} It selects an arbitrary order $\overline T$ for $P'(C_k^i)$.

\item Let $\tilde T$ be the concatenation of $ A_k^{i-q}, \dots, A_k^{i-1}$.
For each account that appears as source in transactions of $\overline T$, consider the last balance
from $\tilde T$ and from the balances provided by the proofs of transactions in $\overline T$.

\item\label{step:conf:rulecheck} It executes $\overline T$ and checks that 
the resulting balance of each transaction fulfills the
non-negative balance rule. Transactions the do not fulfill this rule are
discarded. The resulting $A_k^i$ is derived from $\overline T$ where discarded
transactions are omitted.

\end{enumerate}

Transactions in $A_k^i$ should be considered \emph{confirmed} (or
\emph{accepted}) in the sense they have passed all checks to be inserted in $B^i$. 
To allow the confirmation committees of subsequent rounds to perform time-updating, 
$A_k^i$ is made available to $ C_k^{i+1},  \dots, C_k^{i+q} $ and also to other committees, as
explained in the following.

The sequence of accepted transactions for that round is denoted
$A^i=\bigcup_k A_k^i$, where $A^i$ is an arbitrary sequence that respects the order of
each $A_k^i$.

Even if $B^i$ is yet to be computed, storage nodes can receive from $C_k^i$
state changes that will be part of $B^i$. Transactions in $A_k^i$ are
selectively sent to the storage nodes that need it, to update the part of the
state they manage.

The actual computation of $B^i = B_{i+q-1}$ requires to compute its state root-hash, which means 
computing all the hashes of the conceptual Merkle tree $W$ related to the whole state space.
This is performed by $N_R$ committees, called
\emph{Root-hash Pipeline Committees} (\emph{RPCes}). Each RPC is associated to a part
of $W$ as shown in Figure~\ref{fig:WpartitionedByRPCs}, called the \emph{underlying tree} of the RPC. 
Each underlying tree is rooted to a node whose hash is named \emph{sub-root-hash}.
RPCes themselves form a tree
denoted by $W_\textrm{RPC}$, whose internal nodes have $2^k$ children, 
with the exception of the root that have \emph{at most} $2^k$ children. Dimensioning of $k$, and other 
parameters, is described in Section~\ref{sec:scalability}.
Upon state changes, each RPC is responsible to 
compute all hashes for its underlying tree. There are two kinds of
RPCes. \emph{Leaf} RPCes, that are the leaves of $W_\textrm{RPC}$, and \emph{inner} RPCes, 
that are all other RPCes. Each leaf RPC is responsible for the interval of contiguous addresses
that are leaves of its underlying tree. Since most addresses are unused,
leaf RPCes consider a pruned version of the underlying tree containing only the paths from used addresses
to its root. The underlying tree of each \emph{inner} RPC is a complete binary tree (see Section~\ref{sec:scalability}). 
RPCes are partitioned in \emph{levels} numbered from
$1$ to $h$. Level $1$ contains all leaf RPCes. Level $h$ contains only the
root of $W_\textrm{RPC}$. Each level is one stage of the pipeline.
Hence, the total number of stages of the pipeline (comprising CCes) is $q=h+1$.
Each RPC at level $i<h$ computes a sub-root-hash  that is fed as input to its parent in $W_\textrm{RPC}$. The root
of $W_\textrm{RPC}$ outputs and broadcasts the new block with the corresponding state root-hash.
Theorem~\ref{th:scalability} of Section~\ref{sec:scalability} states that it is
possible to dimension $W_\textrm{RPC}$ and the underlying trees of leaf and inner RPCes to ensure
scalability. 

We write $R^{i+1}_m$ to denote a generic leaf RPC that runs in the $(i+1)$-th round and is \emph{responsible} for the $ m $-th portion of the addresses space.
As mentioned above, leaf RPCes constitutes the second stage of our pipeline and have to receive as input $ A^i $, which is the output of the first stage. However, a leaf RPC does not need the whole $A^i$. 
More in detail, each leaf RPC $ R^{i+1}_m$ receives all and only the transactions in $ A^i $ that modify the balance of an address for which $ R^{i+1}_m$ is responsible. 
In our cryptocurrency model, a single confirmed transaction is sent to two leaf RPCes.
If $ x $ is  a transaction in $ A_k^{i} $,  $ C_k^{i} $ sends $ x $ to $ R^{i+1}_m $ only if the source or the destination of $ x $ is an address for which $ R^{i+1}_m $ is responsible. We denote with $ S_m(A^i) \subseteq A^{i}$ the set of transactions in $A^i$ that 
have one of its involved addresses in the $m$-th portion of the address space and have to be 
received by $ R^{i+1}_m $.

The task performed by $R^{i+1}_m$ is to calculate the sub-root-hash of its underlying tree related to $ B^i $. To do this, it needs the status of its underlying tree related to $ B^{i-1} $.
Since, proofs attached to transactions in $ S_m(A^i) $ are related to $ B_{i-2}=B^{i-q-1} $, they cannot be used alone
to compute all hashes of the underlying tree related to $ B^{i-1} $. In fact, they might have been outdated by transactions 
accepted in $ A^{i-q}, \dots, A^{i-1} $ for which the corresponding block is not yet available.
Hence, each $ R^{i+1}_m $ should also be aware of $ S_m(A^{i-q}),\dots, S_m(A^{i-1}) $.
Each $ R^{i+1}_m $ considers the proofs of these transactions according to their order 
to calculate all hashes of the pruned underlying tree related to $ B^{i-1} $. We call 
\emph{time-shifting} this process (a similar process is described in~\cite{bpp-bmdhtdvfss-19}).
To allow the leaf RPCes of subsequent rounds to perform time-shifting, 
each $C^i_k$ sends $ S_m(A^i_k) $ to $ R_m^{i+1}, \dots, R_m^{i+q+1} $, as well. 

\begin{sixpages}
We provide a formal correctness proof in~\cite{scaling-arXiv}.
\end{sixpages}

\begin{longvers}

\section{Correctness}\label{sec:correctness}
In this section, we formally prove the correctness of the architecture introduced in Section~\ref{sec:solution}.

The following lemma state the correctness of Algorithm~\ref{algo:confirmation}
when run on only one committee.

\begin{Mylemma}[Correctness of the confirmation algorithm]\label{lemma:correctness_cc}
	Algorithm~\ref{algo:confirmation} never returns a sequence that entails
	a violation of the non-negative balance rule.
\end{Mylemma}

\begin{proof}
By construction of the result in Step~\ref{step:conf:rulecheck} of Algorithm~\ref{algo:confirmation}.
\end{proof}

\begin{Mytheorem}[Correctness]
	Given a set of transactions $P^i$ processed, at round $i$, by confirmation
	committees $C_k^i$ producing accepted transactions sequences $A_k^i$, the
	following statements are true.
	\begin{enumerate}

		\item\label{stmt:nonnegative} In any sequence $A^i=\bigcup_k A_k^i$ such that
		$A^i$ respects the order of each $A_k^i$, the non-negative balance rule is
		respected.

		\item\label{stmt:roothash} The state root-hash of $B^i=B_{i+q-1}$ is the root-hash of the new 
		state after the application of $A^i$.

		\item\label{stmt:proofsFromStorage} Storage nodes knows the proofs of the
		addresses they store.
	\end{enumerate}
\end{Mytheorem}
\begin{proof}
 Concerning Statement~\ref{stmt:nonnegative}, observe that $A_k^i$ satisfy the
 non-negative balance rule (Lemma~\ref{lemma:correctness_cc}) and their order is preserved by
 hypothesis in $A^i$. Since for each $k$ the addresses charged in $A_k^i$ are
 not charged in any $A_j^i$ with $j\neq k$, the statement follows.
 
 Concerning Statement~\ref{stmt:roothash}, note that each leaf RPC $R^{i+1}_m$
 considers all the transactions that involve addresses for which $R^{i+1}_m$ is
 responsible that are present in sequences $S_m(A^{i-q}),\dots,S_m(A^{i-1})$, respecting their
 order. RPC $R^{i+1}_m$ can correctly compute its sub-root-hash to pass to its
 parent RPC. In fact, if an internal node of its underlying tree is involved in
 a transaction, $R_m^{i+1}$ receives the proofs attached with the transaction. If an
 internal node of its underlying tree is not involved in any transaction either it
 is pruned or it is a root of a pruned tree. In the first case, $R_m^{i+1}$ does
 not need it. In the latter case, $R_m^{i+1}$ receives its hash in one of the proofs
 available to it. Since, internal RPCes always receive, form their
 children, the hash values for all the leaves of their underlying tree, computing
 their sub-root-hash is trivial. Hence, the statement follows.
 
 Concerning Statement~\ref{stmt:proofsFromStorage}, note that RPCes compute the
 root-hash on the basis of a pruned version $W'$ of $W$, where leaves of $W$ kept in $W'$
 are all used addresses $U$. Each storage node $n$ stores a pruned version $W_n$
 of $W$, where leaves of $W_n$ are all addresses $U_n$ that $n$ intends to
 store. Since $U_n \subseteq U$, also  $W_n \subseteq W'$. Hence, all
 sub-root-hash of pruned subtrees in $W_n$ are known to one of the RPCes, which
 can communicate it to $n$. 
\end{proof}

\end{longvers}

\section{Scalability}\label{sec:scalability} 

In this section, we formally show the scalability of our approach. For real
systems, the workload is usually characterized probabilistically. For
simplicity, we reason assuming a 
deterministic workload. Indeed, similar arguments hold when reasoning with expected
values. We also assume balance changes in a round are
uniformly distributed across the whole address space.
We start by introducing some assumptions and notation.

We denote by $f$ the
frequency of transactions of the blockchain workload. We denote by $\varDelta$ the duration
of a round. We denote by $m= 2 f \varDelta$ the number of addresses whose
balance changes in each round, assuming transactions involve distinct addresses. 
\begin{longvers}
 Let $\tilde W$ a pruned version of $W$ where only
$m$ leaves are kept, i.e., those related to the addresses that change balance in
one round. We note that there is a level $l$ of $\tilde W$, above which $\tilde
W$ is a complete binary tree. As $m$ grows, the pruned part get smaller and $l$
get closer to the leaves.
\end{longvers}

We denote by $j$ the maximum number of hashes that an RPC can compute in a round. Note that
$j$ is constant, since it depends on the CPU power of committees members. We
denote by $e$ the maximum number of balance changes that a leaf RPC $R$ can
process in one round. Clearly $e$ depends on $j$ and
on how changes are distributed in the address space for which $R$ is
responsible, since this determines the number of nodes of the pruned underlying
tree that $R$ has to deal with. However, by the uniform distribution assumption,
$e$ is the same for all leaf RPCes. 
\begin{longvers}
For simplicity, we assume $j$ to be large enough so that the root of the
underlying tree of leaf RPCes is above level $l$. Hence, the underlying tree $U$
of an inner RPC $R$ is a complete binary tree, as stated in
Section~\ref{sec:architecture}. Let $k$ be the number of the levels of $U$. The
nodes of $U$ are $2^k-1$. The children RPCes of $R$ are $2^{k-1}$. For each round,
the inner RPC has to compute one hash for each node of $U$. 
\end{longvers}
\begin{sixpages}
For simplicity, we assume $j$ to be large enough so that underlying trees of
inner RPCes are complete, as stated in
Section~\ref{sec:architecture}.  
\end{sixpages}
The maximum
number of nodes in the underlying tree of an inner RPC is $\hat{j}=2^{\hat k}-1$, where
$\hat k$ is the largest possible integer such that $\hat{j} \leq j$, or equivalently
$\hat k = \floor*{\log_2 (j+1)} $.
The maximum number of levels of the underlying tree of an RPC is also $\hat k$.
We denote by $S(N,f)$ a blockchain system $S$, with the architecture described
in Section~\ref{sec:solution}, with $N$ nodes and with a workload at frequency
$f$.

\begin{Mylemma}\label{lemma:num_rpc}
  Consider a blockchain system $S(N,f)$, of $N$ nodes with workload $f$. Let $e$
  be the maximum number of balance changes that a leaf RPC can handle per round,
  and $\hat j$ be the maximum number of hashes an inner RPC can compute per
  round. If $S$ is well dimensioned, the number of leaf RPCes is at least $2^{\ceil*{\log_2
  (m/e)}}$ and the number of inner RPCes is at least $ \ceil*{\frac{2^{\ceil*{\log_2
  (m/e)}}-1}{\hat{j}}}$, where $m=2f\varDelta$ and $\varDelta$ is the round duration.
\end{Mylemma}
\begin{sixpages}
A proof of Lemma~\ref{lemma:num_rpc} is provided in~\cite{scaling-arXiv}.
\end{sixpages}
\begin{longvers}
\begin{proof}
Leaf RPCes should at least be $\ceil{m/e}$, however, since they have to be the
leaves of a complete binary tree, underlying to inner RPCes, their number have
to be a power of 2. Hence, we they have to be $2^g$ with $g=\ceil{\log_2
(m/e)}$. Consider the union of the underlying graphs of all inner RPCes $W_I$, its
number of levels is $g$, as well.

As $m$ increases, the resource usage of each leaf RPC
increases. When resources of leaf RPCes are fully used, their number is doubled.
Note that, doublings occur when $m/e$ is a power of 2. \emph{Right after a
doubling}, their resources are half used. Increasing $m$, their resource usage
goes from half of its processing capacity to maximum capacity, which occurs \emph{right
before a doubling}.

When leaf RPCes (at level 1 of $W_{RPC}$) are doubled, also inner RPCes at level 2 of $W_{RPC}$ are
doubled. All inner RPCes have their underlying tree with the maximum number of levels, except for the root of
$W_{RPC}$ (see Figure~\ref{fig:WpartitionedByRPCs}), for which its number of levels is
increased by one at each doubling. This occurs until the levels of the underlying tree of the 
root of $W_{RPC}$ reaches $\hat k$.
After that, $g$ increases by one and a new root with
only one node as underlying tree is added at the top. 
Since, leaf RPCes are $2^g$, the nodes of $W_I$ are $2^g-1$. Hence, the number of inner RPCes is given by $
\ceil*{{(2^g-1)}/{\hat{j}}}$. Note that both numerator and denominator represent
the size of a complete binary tree, with $g$ and $\hat k$ levels, respectively. 
The integer part of the result of the
division is the number of inner RPCes with full-sized underlying tree.
The
reminder is the size of the underlying tree of the root of $W_{RPC}$, which is not 
full-sized, in general. 
\end{proof}
\end{longvers}
The following theorem states the scalability of our approach when nodes and workload are proportionally
increased.

\begin{Mytheorem}[Scalability]\label{th:scalability}
  There exists a well-provisioned blockchain system $S(N,f)$, with $N$
  nodes and workload frequency $f$, such that, for all
  $\alpha>1$ such that $\alpha N$ is integer, it is possible to provide
  a well-provisioned blockchain system $\bar S(\alpha N, \alpha f)$.
  
  The above statement holds under the assumption that the balance changes induced by the
  workload are uniformly distributed across the address space.
\end{Mytheorem}

\begin{proof}
\begin{sixpages}
    (sketch) To prove the statement, we choose $S$ (with load $f$) so that the load of the leaf RPCes is minimum. By Lemma~\ref{lemma:num_rpc}, we 
    derive the number of CCes and RCSes needed
    in $\bar S$ (with load $\alpha f$). Then, we show that the increment of committees is compatible with the
    increment of the nodes from $S$ to $\bar S$. We also observe that the
    number of messages sent to and received by each committee is bounded. See~\cite{scaling-arXiv} for a detailed proof.
\end{sixpages}
\begin{longvers}
  We choose as $S(N,f)$ a well-provisioned system, right after a
  doubling, i.e., with minimum resource usage for the leaf RPCes. 
  In our setting, committees of $S$ and $\bar S$ have the same processing
  capabilities. We want to prove that $\alpha N$ nodes in $\bar S$
  are enough for the number of committees needed to process a workload $\alpha f$. We
  first derive the needed number of CCes in $\bar S$ compared to that of $S$ and
  then we do the same for RPCes. Then, we show that they are compatible with the
  increment of the nodes.

  A workload at frequency $f$, generates $f \varDelta$ transactions per round.
  Let $ N_C $ be the number of CCes in $S$. Since $S$ is well-provisioned, each
  CC is able to process $f \varDelta /N_C$ transactions per round. The load of
  $\bar S$ is $\alpha f$, hence, with $\alpha N_C$ CCes, we obtain in $\bar S$ the same
  resource usage of each CC as in $S$. 
  
  Note that, the fact that each $CC$ have to re-process
  transactions accepted in a constant number of previous rounds does not impact
  this reasoning. This is true also for the following argument about RPCes.
  
  We use symbols, $m$, $e$, $\hat j$, and $g$ with the same meaning as before.
  To simplify the proof, in accordance with Lemma~\ref{lemma:num_rpc}, we choose
  to provision $S$ and $\bar S$ with $\ceil*{\frac{2^{\ceil*{\log_2
  (m/e)}}}{\hat{j}}}$ inner RPCes, possibly leaving one of the inner RPCes
  without workload.  Note that, $S$ is well-provisioned and with minimum resource usage of leaf RPC, i.e.
  with $f$ right after a doubling. Hence, in our case, $m/e$ is a power of two, $g=\log_2
  (m/e) +1$, and by Lemma~\ref{lemma:num_rpc} and the above choice for the number of inner RPCes,
  the total number of RPCes in $S$ is $N_R =2^g + \ceil*{{2^g}/{\hat{j}}}= 2m/e +
  \ceil*{2m/(e\hat j)} $. We now consider the total number $\bar N_R$ of RPCes
  needed by $\bar S(\alpha N, \alpha f)$. For $\alpha=2^t$ with $t$ positive integer, it
  should be $\bar N_R  = 2 \alpha m/e+ \ceil*{2 \alpha m/(e\hat j)}$. Since, 
  $\ceil*{2 \alpha m/(e\hat j)}\leq \alpha \ceil*{2 m/(e\hat j)}$, we have that $\bar N_R
  \leq \alpha N_R$. Hence, $\alpha N$ nodes are enough for $\bar S$ to be
  well-provisioned. Our assumptions imply that $\bar S(\alpha N, \alpha f)$ is
  again right after a doubling and with only half of the resources of leaf RPCes
  used. Hence, with the same $\bar N_R$, $\bar S$ is well provisioned till the
  next doubling, which covers all the values of $\alpha$ such that, $2^t<\alpha<2^{t+1}$.
  
  Hence, the needed increment of the number of CCes and of the number of
  RPCs from $S$ to $\bar S$ is at most by a factor of $\alpha$, and $\bar S$ has $\alpha N$
  nodes. This proves that $\bar S$ is well-provisioned regarding processing aspects.

  Additionally, we note that the messages sent and received by each committee in
  $\bar S$  in each round is no more that the double of the number of messages
  sent and received in $S$. Further, we note that the storage can be realized so that
  the number of storage nodes that store a certain account is bounded.
  
  The above observations complete the proof of the statement.
\end{longvers}
\end{proof}

\section{Discussion and Future Works}\label{sec:discussion}

In this section, we discuss the effectiveness or our approach and certain
aspects that are not analyzed in the rest of the paper. In particular, it is
important to understand if the proposed approach is a better solution to the
blockchain trilemma than the previously known ones. This means understanding if
scaling requires to limit security and/or decentralization, since the
scalability of our appraoch has been formally stated in
Section~\ref{sec:scalability}.

\paragraph{Decentralization}
Concerning decentralization, note that in our system all nodes
cooperate in the creation of a new block.  Further, even if the committees do
not have all the same role, we can assign nodes randomly to each committee, possibly changing them periodically (like,
for example, in~\cite{chen2019algorand}). In this way, the role of the nodes is statistically homogeneous. 
Due to these considerations, we
think that scaling in our approach does not affect decentralization.

\paragraph{Security}
Considering security, many other research works and practical systems relay on
the security of a consensus algorithm run by a restricted set of nodes forming a
committee. In this setting, many attacks require the attacker to control the
majority of the committee members. However, when committees are randomly selected, this becomes harder 
as the number of nodes increases. In conjunction with a proof-of-stake approach (for example)
this protects against Sybil attacks.
If consensus algorithm is robust enough, the presence of several committees has
a negligible impact on security.
In this sense, security of our approach increases when scaling to a higher number of nodes.
Clearly, security is about many other aspects, but the difficulty to subvert the
consensus is usually considered in the context of discussions about the
blockchain trilemma.

It is worth to note that for the correctness of our approach nodes need to  keep only a constant amount 
of blocks. However, for security reasons, nodes can keep more blocks (or other related information).
For example, in Vault~\cite{leung2019vault} the join of a new node is securely performed without downloading 
the whole blockchain. Analogous approaches can be adopted in our context.

A security analysis with respect to a formally stated threat model is leaved as future work.

\paragraph{Committee members selection} Our approach is applicable independently from the
way members of each committee are selected. 
Their selection can be done using a public shared source of
randomness or using verifiable random functions, as in~\cite{chen2019algorand}.
However, there is a caveat regarding this in our approach. Since intermediate
results of the pipeline are passed to committees that need them in the next 
rounds, if members of committees change, these have to be decided and published
before data is sent to them. Note that, resorting to broadcast is not possible
since this would impair scalability.

\paragraph{Inter-committee communication} In Section~\ref{sec:solution}, we
often relayed on the possibility for a committee to communicate data to other
committees that need them in the next rounds. Inter-committee communications
should be part of the consensus protocol, in the sense that each receiver should
accept a message $m$ from a sender committee $S$ only after having checked that 
$m$ was sent by a number of members above a certain threshold.

\paragraph{Network communications} In our architecture, most messages are sent from nodes to
a bounded number of other nodes (members of a certain committee or storage nodes storing a certain account).
\begin{sixpages}
We think that a multicast approach can be used to fulfill our needs.
The formal statement of the requirements of this underlying multicast layer as well as its design 
is leaved as future work. Some further considerations and references to related networking literature 
can be found in~\cite{scaling-arXiv}.
\end{sixpages} 
\begin{longvers}
We note that, using unicast communications, the number of actual messages turns out to be
quadratic in the size of the committees. While practically this might be a
problem, from a theoretical point of view this is not the case, since the size
of committees is constant. However, a critical aspect is that unicast
communications require the destination to be known, which is not easy to obtain
in a scalable manner when the committees change regularly.

We think that a multicast approach can be used to fulfill our needs.
The formal statement of the requirements of this underlying multicast layer as well as its design 
is leaved as future work. However, it may be worth to mention that the use of
standard multicast techniques may not completely suite our needs. In particular,  
the following aspects should be considered when adopting multicast for inter-committee communications or for submission of new transactions to the proper confirmation committee.
\begin{enumerate}
   \item The members of a multicast group might change rapidly,
   depending on the round duration.
   \item The preparation of the multicast groups can be performed in advance with
   respect to when they are needed, even if this requires to anticipate the disclosure of 
   committee members.
   \item The multicast groups are needed for only one round and then discarded.
   This might simplify the development of a specific technique for this
   application.
\end{enumerate}
Regarding the use of multicast for messages destined to storage nodes, we note that 
the number of needed multicast channels might very large: one 
for each used address.

Many research works about scalable multicast routing are
available in literature (see, for
example,~\cite{tian1998forwarding,gronvall2002scalable,wong2000analysis,tapolcai2014optimal,balay2012scalable,broder2004network}). Fast-join multicast routing was studied in~\cite{cho2003fjm}.
\end{longvers}

\paragraph{Synchronization and committee decision failing} 
\begin{sixpages}
We assumed a sort of global synchronization, which is impractical. 
Further, in practice, a committee may occasionally fail and not produce its output.
We leave as a future work 
the investigation of these aspects.
\end{sixpages}
\begin{longvers}
In our description,
we essentially assumed a sort of global synchronization.
In practice,
synchronization spread across a large number of nodes (although partitioned in
committees) might be difficult to achieve. The problem of modifying our approach
to relax synchronization requirements is left as a future work. A related
problem is the failure of a committee to reach an agreement, which is unlikely
to for a direct attack, but may occur under large network faults. Again, we left
the investigation of these aspect as a future work.
\end{longvers}

\section{Conclusions} 
\label{sec:conclusions}

We showed a novel blockchain design that distributes the burden to create the next block on 
many parallel executing committees (involving ``almost'' all nodes) 
and that avoids broadcast in all cases that are critical for scalability.
\begin{sixpages}
We formally stated the scalability of our approach and discussed how this does not impair decentralization and security. 
\end{sixpages}
\begin{longvers}
We provided formal proof of the scalability of our approach and of its correctness.
We also discussed how scaling does not impair decentralization and security. 
\end{longvers} 
Hence, our architecture can be regarded as a 
solution to the blockchain scalability trilemma, in the studied setting.
\begin{longvers}

Some future works were already mentioned in Section~\ref{sec:discussion}.
In particular, from a theoretical point of view, a formal security proof is needed, as well as a further investigation of synchronization and behavior under committee consensus failure.
Further, since we assumed to have a scalable multicast protocol, it should be interesting 
to understand how to realize this protocol and how it is possible to rely on current technology to realize an efficient one.
From a practical point of view, an experimentation or simulation with realistic parameters 
would be desirable. 
\end{longvers}

\begin{longvers}
\section{Acknowledgments}
We are extremely grateful to Ciro Oliviero for his important contribution in the beginning of this research.
\end{longvers}

\end{document}